\documentclass[11pt,letterpaper]{article}
\usepackage[utf8]{inputenc}

\usepackage[export]{adjustbox}
\usepackage{algorithm}
\usepackage{algorithmicx,algpseudocode}
\usepackage{amsmath,amssymb,amsfonts,amsthm}
\usepackage{array}
\usepackage{bibunits}
\usepackage{booktabs}
\usepackage{caption}
\usepackage{cite}
\usepackage{colortbl}
\usepackage{csvsimple}
\usepackage{etoolbox}
\usepackage[T1]{fontenc} 
\usepackage[margin=1in]{geometry}
\usepackage{graphicx}
\usepackage{glossaries,glossaries-extra}
\usepackage[hypertexnames=false]{hyperref}
\usepackage{cleveref}
\usepackage{ifthen}
\usepackage{import}
\usepackage{listings}
\usepackage{longtable}
\usepackage{makecell}
\usepackage{makecmds}
\usepackage{multirow}
\usepackage{natbib}
\usepackage[section]{placeins}
\usepackage{pythonhighlight}
\usepackage{rotating}
\usepackage[moderate]{savetrees}
\usepackage[mark=***]{sectionbreak}
\usepackage{setspace}
\usepackage{siunitx}
\usepackage{stringstrings}
\usepackage{subcaption}
\usepackage{subcaption}
\usepackage{tabularx}
\usepackage{textcomp}
\usepackage{vcell}
\usepackage{xcolor}
\usepackage{xfp}
\usepackage{xparse}
\usepackage{xstring}

\newcommand{\ceil}[1]{\left\lceil #1 \right\rceil}

\makeatletter
\let\pragma@iinput=\@iinput
\def\@iinput#1{\xdef\@pragmafile{#1}\pragma@iinput{#1} }
\def\@pragmafile{default}
\def\pragmaonce{%
   \csname pragma@\@pragmafile\endcsname
   \global\expandafter\let \csname pragma@\@pragmafile\endcsname =  
}
\makeatother

\Crefname{ALC@unique}{Line}{Lines}

\ifdefined\mydraft
\mydraft
\fi

\defaultbibliography{bibl}
\defaultbibliographystyle{apalike-ejor}

\graphicspath{{./}{./hereditary-stratigraph-concept/tex}{./hstrat-evolutionary-inference}}

\newtheorem{theorem}{Theorem}

\makeglossaries

\newglossaryentry{differentia}
{
    name=differentia,
    description={randomly generated information that can be used to differentiate two strata at the same layer (with a probability of spurious collision)}
}

\newglossaryentry{stratum}
{
    name=stratum,
    description={differentia container associated with a particular historical stage}
}

\newglossaryentry{deposit}
{
    name=deposit,
    description={the act of extending the historical record with a new stratum}
}

\newglossaryentry{deposition}
{
    name=deposition,
    description={the stratum appended to the historical record}
}

\newglossaryentry{time point}
{
    name=time point,
    description={refers to the stage where a specific number of depositions have taken place}
}

\newglossaryentry{deposition time}
{
    name=deposition time,
    description={refers to the time point at which a stratum was deposited, zero indexed}
}

\newglossaryentry{genesis}
{
    name=genesis,
    description={the time point associated with the first stratum deposition}
}

\newglossaryentry{time}
{
    name=time,
    description={number of depositions elapsed between two time points}
}

\newglossaryentry{layer age}
{
    name=layer age,
    description={the number of depositions elapsed since a layer's deposition time}
}

\newglossaryentry{record depth}
{
    name=record depth,
    description={the number of depositions elapsed onto the historical record --- the number of layers within a historical record}
}

\newglossaryentry{layer}
{
    name=layer,
    description={position within a historical record associated with a single time point, which may or may be occupied}
}

\newglossaryentry{retained/pruned layer}
{
    name=retained/pruned layer,
    description={a layer within a historical record with present/removed strata, respectively}
}

\newglossaryentry{layer time point}
{
    name=layer time point,
    description={the deposition time associated with a layer}
}

\newglossaryentry{pruning}
{
    name=pruning,
    description={deletion of strata from a historical record (i.e., to reduce space occupied by the record) --- also used to refer to deletion of perfect tracking records for extinct lineages,}
}

\newglossaryentry{retention}
{
    name=retention,
    description={the act of carrying over a stratum into the next time point during the stratum deposition process}
}

\newglossaryentry{gap}
{
    name=gap,
    description={layers associated with contiguous time points that have been pruned ---- introduces inference uncertainty when comparing two columns}
}

\newglossaryentry{gap width}
{
    name=gap width,
    description={the number of contiguous time points that have been pruned --- gap with increases inference uncertainty}
}

\newglossaryentry{sparse/dense retention}
{
    name=sparse/dense retention,
    description={refers to relatively wide or relatively tight gap width, respectively}
}

\newglossaryentry{gap width distribution}
{
    name=gap width distribution,
    description={how gap widths relate to layer deposition times}
}

\newglossaryentry{resolution}
{
    name=resolution,
    description={the width of the gap containing a pruned layer or immediately following a retained layer (may be zero in the case of two successive retained strata)}
}

\newglossaryentry{stratum retention policy algorithm}
{
    name=stratum retention policy algorithm,
    description={the decision-making procedure of which strata to prune at each time point (also referred to simply as ``policy'')}
}

\newglossaryentry{policy resolution guarantee}
{
    name=policy resolution guarantee,
    description={upper bounds on resolutions across layers of a historical record with respect to layer age and/or record depth}
}

\newglossaryentry{extant record size}
{
    name=extant record size,
    description={the quantity of strata retained within a historical record at a particular time point}
}

\newglossaryentry{extant record order of growth}
{
    name=extant record order of growth,
    description={the asymptotic scaling relationship between the extant record size and record depth; see Section \ref{sec:extant_record_oog}}
}

\newglossaryentry{pruning enumeration}
{
    name=pruning enumeration,
    description={calculation of the set of strata to be pruned at a particular time point under a retention policy}
}

\newglossaryentry{policy enactment}
{
    name=policy enactment,
    description={the act of performing pruning enumeration and deleting strata with enumerated deposition times}
}

\newglossaryentry{update}
{
    name=update,
    description={the process of performing a deposition and applying policy enactment}
}

\newglossaryentry{update time complexity}
{
    name=update time complexity,
    description={the scaling relationship associated with the number of computational operations necessary to perform an update}
}

\newglossaryentry{founding stratum}
{
    name=founding stratum,
    description={the first stratum deposited into the historical record, the oldest stratum}
}

\newglossaryentry{newest stratum}
{
    name=newest stratum,
    description={the most recent stratum deposited into the historical record}
}

\newglossaryentry{extant record}
{
    name=extant record,
    description={the set of strata that have been retained through the policy at the present time point}
}

\newglossaryentry{extant record enumeration}
{
    name=extant record enumeration,
    description={calculation of deposition times present in the extant record at a time point under a retention policy}
}

\newglossaryentry{policy self-consistency}
{
    name=policy self-consistency,
    description={the requirement for each deposition time within an extant record enumeration to be consistently present in all extant enumerations since that deposition time}
}

\newglossaryentry{historical record}
{
    name=historical record,
    description={refers to the set of layers defined up to the current time point, also referred to as a ``record''}
}

\newglossaryentry{hereditary stratigraphic column}
{
    name=hereditary stratigraphic column,
    description={container for a historical record --- in phylogenetic applications, associated with a digital population member, also referred to as a ``column''}
}

\newglossaryentry{annotation}
{
    name=annotation,
    description={the one-to-one association of hereditary stratigraphic records with individual digital organisms to facilitate phylogenetic analysis}
}

\newglossaryentry{inheritance}
{
    name=inheritance,
    description={the act of copying the parent organism's hereditary stratigraphic column annotation and performing an update to create the offspring organism's hereditary stratigraphic column annotation during a reproduction event}
}

\newglossaryentry{inference}
{
    name=inference,
    description={best-effort estimation of historical phylogenetic relationships from extant hereditary stratigraphic columns}
}

\newglossaryentry{perfect tracking}
{
    name=perfect tracking,
    description={maintenance of an exact record of phylogenetic events during an evolutionary simulation}
}

\newglossaryentry{stream curation problem}
{
    name=stream curation problem,
    description={poses the question of how to satisfactorily maintain a temporally representative collection of stored observations on a rolling basis as new observations stream in}
}

\begin{document}

\title{ Algorithms for Efficient, Compact Online Data Stream Curation }
\author{
    Matthew Andres Moreno\textsuperscript{1,2,3}\thanks{Corresponding author: \texttt{morenoma@umich.edu}} \quad
    Santiago Rodriguez Papa\textsuperscript{4} \quad
    Emily Dolson\textsuperscript{4,5} \quad
}
\date{}

\newcommand{\affil}[1]{\textsuperscript{#1}}
\newcommand{\affiliations}{
\affil{1} Ecology and Evolutionary Biology, University of Michigan, Ann Arbor, United States \\
\affil{2} Center for the Study of Complex Systems, University of Michigan, Ann Arbor, United States \\
\affil{3} Michigan Institute for Data Science, University of Michigan, Ann Arbor, United States
\affil{4} Department of Computer Science and Engineering, Michigan State University, East Lansing, United States \\
\affil{5} Ecology, Evolution, and Behavior, Michigan State University, East Lansing, United States \\
}

\maketitle

\begin{center}
\affiliations
\end{center}

\begin{abstract}
Data stream algorithms tackle operations on high-volume sequences of read-once data items.
Data stream scenarios include inherently real-time systems like sensor networks and financial markets.
They also arise in purely-computational scenarios like ordered traversal of big data or long-running iterative simulations.
In this work, we develop methods to maintain running archives of stream data that are temporally representative, a task we call ``stream curation.''
Our approach contributes to rich existing literature on data stream binning, which we extend by providing stateless (i.e., non-iterative) curation schemes that enable key optimizations to trim archive storage overhead and streamline processing of incoming observations.
We also broaden support to cover new trade-offs between curated archive size and temporal coverage.
We present a suite of five stream curation algorithms that span $\mathcal{O}(n)$, $\mathcal{O}(\log n)$, and $\mathcal{O}(1)$ orders of growth for retained data items.
Within each order of growth, algorithms are provided to maintain even coverage across history or bias coverage toward more recent time points.
More broadly, memory-efficient stream curation can boost the data stream mining capabilities of low-grade hardware in roles such as sensor nodes and data logging devices.
\end{abstract}

\clearpage
\newpage

\begin{bibunit}

\section{Introduction} \label{sec:introduction}

Absent indefinite storage capacity, any piece of incoming streaming data must eventually be either evicted or digested if space is to be made available for new input \citep{gaber2005mining}.
This constraint is a crucial consideration in algorithm design for data streams, scenarios involving read-once inputs available only in a strictly ordered sequence.
Such streams' ordering may be dictated by inherently real-time processes (e.g., sensor readings) or retrieval limitations of storage media (e.g., a tape archive) \citep{henzinger1998computing}.
The data streaming model assumes input greatly exceeds memory capacity, with many analyses simply treating streams as unbounded \citep{jiang2006research}.

Data streaming scenarios pervade domains across science and industry \citep{aggarwal2009data,akidau2015dataflow}.
Commercial application areas include sensor networks \citep{elnahrawy2003research}, big-data analytics \citep{he2010comet}, real-time network traffic analysis \citep{johnson2005streams,muthukrishnan2005data}, systems administration \citep{fischer2012real}, and financial analytics for fraud prevention and algorithmic trading \citep{rajeshwari2016real,agarwal2009faster}.
Notable scientific applications arise in environmental/climate monitoring \citep{hill2009real} and astronomy \citep{graham2012data}.
Purely-programmatic computation can also behave as a data stream --- iterative simulation processes traverse vast expanses of ephemeral intermediate state that must be traced to verify simulation dynamics and assess simulation outcomes \citep{abdulla2004simulation,schutzel2014stream}.

Indeed, this broad utility has begat an extensive corpus of data stream algorithms.
Common objectives include rolling summary statistic calculations \citep{lin2004continuously}, on-the-fly data clustering \citep{silva2013data}, live anomaly detection \citep{cai2004maids}, and rolling event frequency estimation \citep{manku2002approximate}.
Data stream algorithms typically draw on one or more of three key stratagems: (1) rolling mechanisms, which restrict consideration to a FIFO tranche of recent data, (2) accumulation, which successively folds data into a summary statistic (e.g., sum, count, etc.) where data is repeatedly applied to a fixed amount of memory or resources, and (3) binning, which consolidates data within time interval bins to create a coarsened record.

Here, we focus on the third stratagem, binning.
Specifically, we develop efficient procedures to maintain temporally-representative subsamples of a data stream on a rolling basis.
That is, to read sequential observations from a data stream on an ongoing basis and sequence their disposal to maintain a record of data stream observations.

We term the rolling management of samples subsetted from a data stream as ``stream curation.''
Proposed algorithms span several possible requirements for two curatorial properties: (1) ``order of growth'' --- how curated collection size should grow in proportion to stream depth and (2) ``gap size bounds'' --- how retained samples should be spaced across stream history.
These considerations arise in various combinations across existing work \citep{aggarwal2003framework,han2005stream}, reviewed in detail later on;
here, we systematize these curatorial properties and contribute novel curatorial policy implementations distinguished by efficiency.
Each contributed policy includes indexing schemes that simultaneously support both efficient update operations and efficient storage of retained stream values in a flat array, requiring only $O(1)$ storage overhead --- a single counter value.

Although we do not treat it directly here, the original motivating application for contributed stream curation algorithms is ``hereditary stratigraphy,'' a recently-developed technique for distributed tracking of copy trees among replicating digital artifacts \citep{moreno2022hereditary}.
Applications of such tracking include phylogenetic analysis of highly-distributed genetic algorithms and evolutionary simulations, and provenance analysis of decentralized social network content, peer-to-peer file sharing, and computer viruses.
A brief description of hereditary stratigraphy is instructive to the timbre of algorithms contributed here.

\subsection{Hereditary Stratigraphy}

In order to reconstruct histories of relatedness, hereditary stratigraphy annotates replicating artifacts with a record of checkpoint fingerprints that grows by accretion with each replication event.
Comparing two artifacts' checkpoint records tells the extent of their common ancestry, as annotations will share common fingerprints up through the time of their last common ancestor and then differ.

Considering generational fingerprint records as a data stream, hereditary stratigraphy applies binning techniques to manage fingerprint accretion --- paring down retained fingerprints while maintaining checkpoints spaced across generations back to the progenitor artifact.
In the context of hereditary stratigraphy, stream curation decides how annotation size scales with generations elapsed by controlling how many retained strata accumulate.
Stream curation decisions directly influence capability for ancestry inference, because the onset of lineage divergence can only be discerned where fingerprints are retained.
Requirements on space usage and inferential power differ substantially between use cases of hereditary stratigraphy, so flexible support for a variety of record size/inferential power trade-offs is crucial.

Of particular note, however, is hereditary stratigraphy's necessity for compact representation of fingerprint records.
Because reduced fingerprint size allows more fingerprints to be retained, typical use will take fingerprints as individual bits, or possibly bytes (to avoid addressability complications).
In this context, representational overhead incurred, e.g., by explicitly storing fingerprints' individual stream sequence indices, can easily bloat annotations' footprint severalfold.
For some use cases, annotated artifacts will number millions or higher, so annotation inefficiency may substantially burden memory, storage, and network bandwidth (i.e., serialized artifact-annotation exchange).

Here, however, we present these algorithmic foundations developed for hereditary stratigraphy in the more generalized frame of data stream processing.
We describe a suite of indexing schemes for stream curation that support (1) linear, logarithmic, and constant scaling relationships between record size and generations elapsed and (2) both even-time and recency-biased distributions of retained stream items.
Implementations provided for each drop representational overhead for curated stream data to a single counter value.
Presented algorithms are published through the \texttt{hstrat} Python package for hereditary stratigraphy \citep{moreno2022hstrat}, but can be directly accessed through public APIs fully independent of other aspects of hereditary stratigraphy methodology.

To provide further introduction to key concepts behind stream curation, the next sections situate our proposed stream curation procedures within existing data stream literature and consider applications of stream curation data loggers and sensor networks.

\subsection{Stream Curation} \label{sec:streaming-curation}

Under an iterative model, the passing of time operates something like a ``first-in, nothing-out'' queue --- successive time steps simply pile on ad infinitum.
As time accumulates, each elapsed time step recedes ever deeper.
A discrete event's absolute time point does not change, but its relation to the present does.  
This inevitability is crux to the ``stream curation'' problem, which we establish to describe rolling maintenance of a temporally representative cross-section of data stream observations.

Stream curation algorithms must answer how many observations should be kept at any point in time, but also how observations that are retained should be spaced out over past time.
Appropriate choices vary by use case, and no stream curation policy can meet all possible demands.
For this reason, we consider a spectrum across two factors: size limitation, i.e., how many observations may be retained, and resolution guarantees, i.e., maximum gap sizes.
For each space-vs-resolution trade-off explored, we provide an implementation algorithm meeting criteria of computational reducibility and self-consistency, defined below.
We first discuss stratum curation policy trade-off stipulation criteria and policy implementation for stream curation, then close with connections of stream curation to existing work and potential applications.

\subsubsection{Stream Curation Policy Stipulation}

Two primary dimensions of stream curation policy matter for practical purposes: 1) retained collection size and 2) time gap sizes between retained observations.

We bound retained collection size to a fixed value or a function of time elapsed.
Asymptotic bounds on the scaling relationship between collection size and time are the ``size order of growth.''
In some cases, we also define hard bounds, referred to as a ``size cap,'' in light of practical considerations; for example, a user may have a fixed size memory allocation in which to store a curated collection.

Bounds on spacing between retained observations, a ``resolution guarantee,'' should depend on both the elapsed time and the temporal depth of a particular observation.
Taking into account historical depth allows skew in retained observation density.
For instance, observations may be retained at evenly-spaced time points or, alternately, thinned proportionately to historical depth.
The latter approach biases observational detail to recent time, which may be important in some use cases.

For example, in the context of hereditary stratigraphy, recency-proportional resolution is typically preferable.
Coalescent theory predicts a tendency for evolution-like processes to produce phylogenies with many recent branches and progressively fewer ancient branches \citep{nordborgCoalescentTheory2019, berestyckiRecentProgressCoalescent2009}.
Thus, fine inferential detail over recent time points usually proves more informative to phylogenetic reconstruction than detail over more ancient time points.
Indeed, trials reconstructing known lineages have found that recency-skewing retention provides better quality reconstructions \citep{moreno2022hereditary}.

Note that size bounds and resolution guarantees must hold across all time points for use cases where observation collections will see sustained use over time or the endpoint for an observation collection is indeterminate (e.g., computations with a real-time termination condition).
This factor obliges policy design nuance: if resolution guarantees shift as generations elapse and observations become more ancient, cohorts of retained strata must, in dwindling, morph through a constrained series of retention patterns.

\subsubsection{Stream Curation Policy Algorithms}

A stream curation policy algorithm produces a sequence of retained observation sets, one for each time point when the underlying data stream is sampled.
Policy algorithms must meet several requirements.

First and foremost, each of an algorithm's retention sets should satisfy all stipulated requirements on collection size and gap size.

Additionally, to be viable, each retention set must be a subset of all preceding retention sets.
Otherwise, a previously discarded observation would be selected for inclusion, which is impossible (once data is discarded it cannot be retrieved).
We call this property self-consistency.

For the sake of efficient operation, we impose a final nuts-and-bolts requirement on algorithm implementation: computational reducibility, meaning that observation times retained at any point must be directly enumerable.
This capability enables observations' time points to be deduced positionally from a buffer index, so observation times may be omitted.
In the context of hereditary stratigraphy --- where e.g., observations are single bits or bytes --- several-fold space savings may result.

Memory savings from computational reducibility can matter greatly.
Since the austere early days of computing, typical hardware has trended away from resource scarcity \citep{kushida2015cloud}, yet memory efficiency remains crucial in certain contexts where hardware trends have stagnated or even regressed memory capacity.

Aspects of high-performance computing (HPC) expect to continue scaling out with lean processing cores \citep{sutter2005free,morgenstern2021unparalleled}.
The Cerebras Wafer-Scale Engine (WSE) epitomizes this trend, packaging an astounding 850,000 computing elements onto a single die.
Individual WSE cores, however, have just 48kb of memory and can only communicate within a local mesh \citep{cerebras2021wafer,lauterbach2021path}.

Component economization and miniaturization has also influenced the Internet of Things (IoT) revolution \citep{rfc7228,ojo2018review}, an ongoing march of ubiquitization potentially culminating in a ``smart dust'' of downscale, low-end hardware \citep{warneke2001smart}.
The Michigan Micro Mote platform for instance, provisions a mere 3kb of retentive memory within its cubic millimeter form factor \citep{lee2012modular}.
More recent work has explored devices tucked within dandelion-like parachutes \citep{iyer2022wind}.
That chipset is yet more austere, provisioning 2 kilobytes of volatile flash memory --- and a mere 128 bytes of retentive memory \citep{microchip2014atiny20}.
As engineers continue to plumb the extremities of technical feasibility, bare-bones computing modalities will persist, and necessitate lightweight data stream algorithms such as ours.

\subsubsection{Existing Work Related to Stream Curation}

Stream curation closely relates to existing binning procedures that group together and consolidate contiguous subsections of a data stream.

The fixed-resolution policy algorithm presented in Section \ref{sec:fixed-resolution-algo} is simple down sampling via decimation \citep[p. 31]{crochiere1983multirate}.
Our depth-proportional resolution (Section \ref{sec:depth-proportional-resolution-algo}) and recency-proportional resolution (Section \ref{sec:recency-proportional-resolution-algo}) algorithms share close structural similarity with the online equi- and vari-segmented schemes proposed in
\citep{zhao2005generalized}.
The depth-proportional resolution structure has appeared additionally in ``pyramidal'' and ``tilted'' time window schemes \citep{aggarwal2003framework,han2005stream}.

To our knowledge, these previous implementations all unfold through stateful iteration, with representational overhead for each stored value (e.g., timestamps, segment length values); stateless enumerations of retained set composition are original to our work in this paper.
We are also not aware of existing equivalents or near-equivalents of the presented geometric sequence $n$th root and curbed recency-proportional resolution policy algorithms (Sections \ref{sec:geom-seq-nth-root-algo} and \ref{sec:curbed-recency-proportional-resolution-algo}).

Work on ``amnesic approximation'' tackles a similar end goal, but has only loose technical overlap.
\cite{palpanas2004online} provides a generalized scheme to incrementally down-sample a data stream pursuant to a user-defined time-back-to-value function by means of a stateful iterative process.

\subsubsection{Applications of Stream Curation}

Correspondences between stream curation and more general binning on data streams suggest avenues for application of stream curation policy algorithms across data stream scenarios.

Perhaps most plainly, the stream curation down-sampling problem parallels those faced by unattended data logger devices that manage incoming observation streams, often on an indefinite or indeterminate basis.
Devices incorporated into wireless sensor networks may also experience irregular device uplink schedules.
The ``mobile sink'' paradigm \citep{jain2022survey}, for example, relies on network base station(s) that physically traverse the coverage area and uplink sensor nodes on potentially sporadic patrol schedules.

Existing work has largely applied rolling full retention of most recent data within available buffer space \citep{fincham1995use} or dismissal of incoming data after storage reaches capacity \citep{saunders1989portable,mahzan2017design}.
Strategies to maintain a cross-sectional time sample appear scant, although there has been some work to extend the record capacity of data loggers through application-specific online compression algorithms \citep{hadiatna2016design}.

\sectionbreak

The remainder of this paper surveys a suite of stream curation algorithms, introducing intuition, presenting the formal definition, proving self-consistent stratum discard sequencing, and demonstrating resolution and collection size properties.
For reference, we include a Glossary of terminology related to hereditary stratigraphy and stream curation in the Appendix.


\section{Stream Curation Algorithms} \label{sec:annotation-algorithms}


\providecommand{\dissertationelse}[2]{%
\ifdefined\DISSERTATION
#1
\else
#2
\fi
}


\providecommand{\dissertationonly}[1]{%
\ifdefined\DISSERTATION%
#1%
\else%
\fi
}

\begin{figure*}
  \centering
  \footnotesize
  \begin{tabular}{m{0.07\textwidth}@{}|c@{}|c@{\hskip 0.01\textwidth}|m{0.14\textwidth}}
\hspace{-1ex}Policy&Lower-Resolution Parameterization&Higher-Resolution Parameterization&\makecell[c]{Properties}\\\hline
    \rotatebox{90}{\textbf{Fixed Resolution}}
  &
    \makecell{
      \includegraphics[valign=t,width=0.3\textwidth]{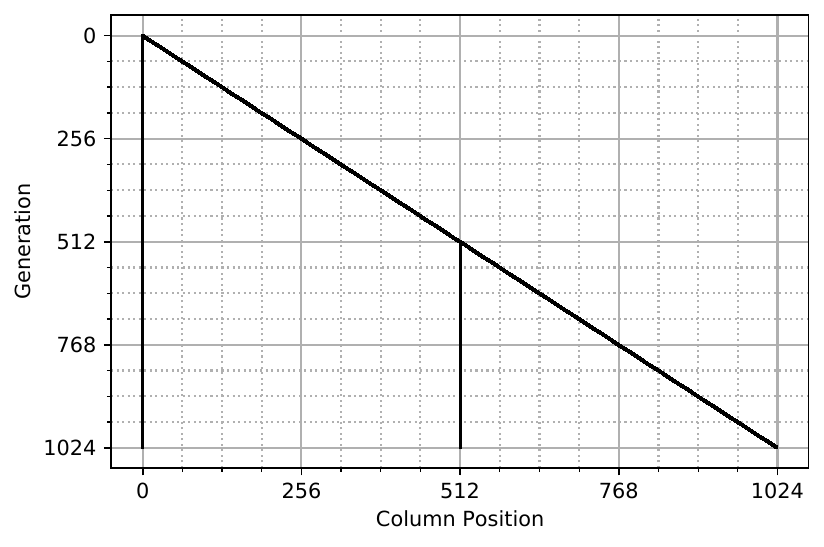}
    }
  &
    \makecell{
      \includegraphics[valign=t,width=0.3\textwidth]{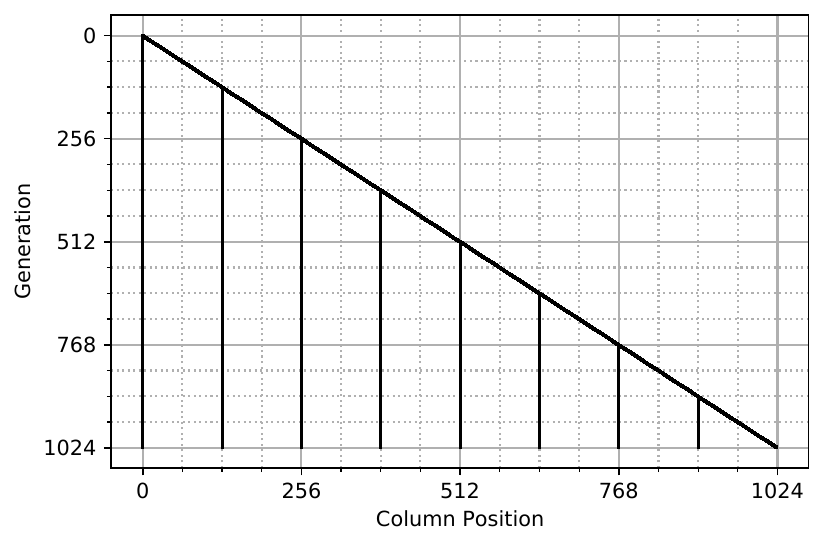}
    }
  &
  \makecell[{{p{0.14\textwidth}}}]{
  \centering
    \bf{Space Complexity}\\
    $\mathcal{O}(n)$\\
    \bf{MRCA Uncertainty}\\
    $\mathcal{O}(1)$
  }
  \makecell[{{p{0.14\textwidth}}}]{
  \raggedright
    where $n$ is gens elapsed.
  }\\\hline
    \adjustbox{
      minipage=10em,
      rotate=90,
    }{
      \centering
      \textbf{Depth-Proportional\\Resolution}
      \par
    }
  &
    \makecell{
      \includegraphics[valign=t,width=0.3\textwidth]{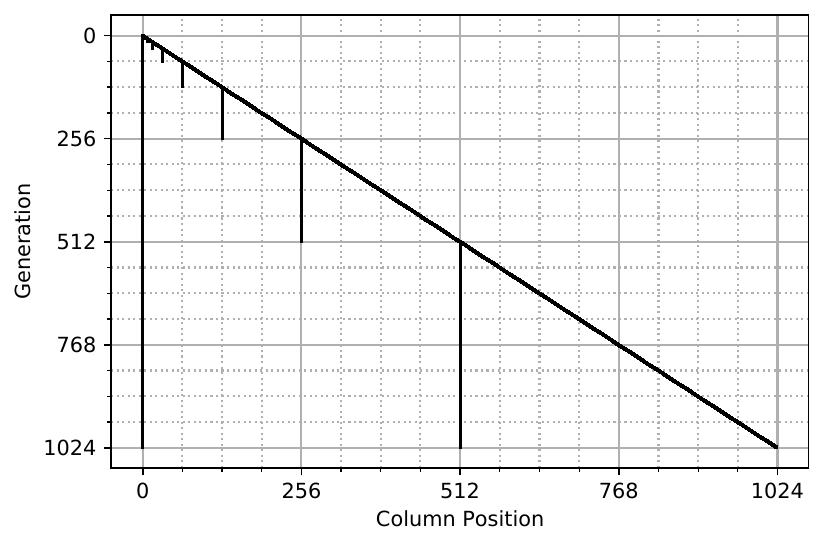}
    }
  &
    \makecell{
      \includegraphics[valign=t,width=0.3\textwidth]{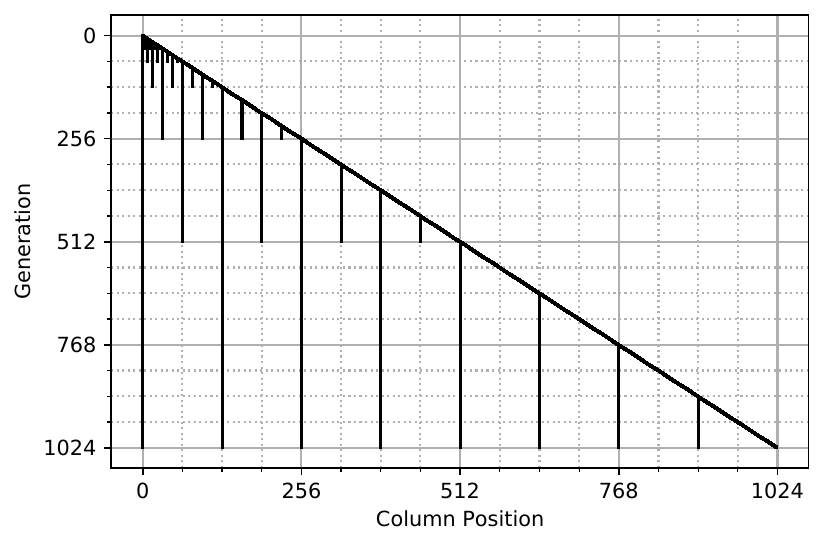}
    }
  &
  \makecell[{{p{0.14\textwidth}}}]{
  \centering
    \bf{Space Complexity}\\
    $\mathcal{O}(1)$\\
    \bf{MRCA Uncertainty}\\
    $\mathcal{O}(n)$
  }
  \makecell[{{p{0.14\textwidth}}}]{
  \raggedright
    where $n$ is gens elapsed.
  }\\\hline
  \adjustbox{
    minipage=10em,
    rotate=90,
  }{
    \centering
    \textbf{Recency-proportional\\Resolution}
    \par
  }
  &
  \makecell{
    \includegraphics[valign=t,width=0.3\textwidth]{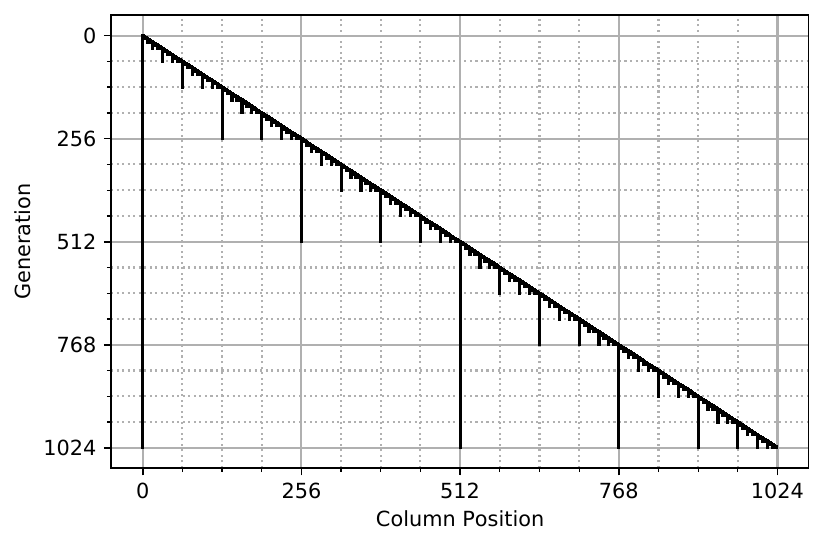}
  }
  &
  \makecell{
    \includegraphics[valign=t,width=0.3\textwidth]{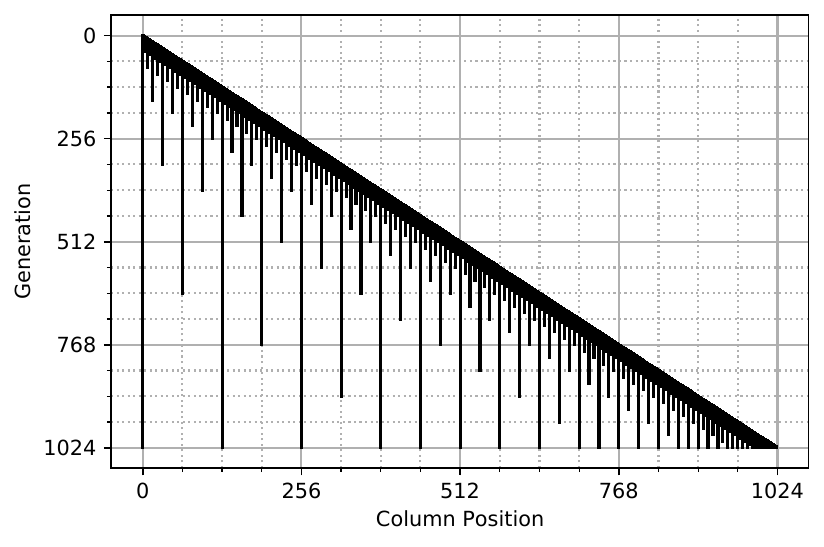}
  }
  &
  \makecell[{{p{0.14\textwidth}}}]{
  \centering
    \bf{Space Complexity}\\
    $\mathcal{O}(\log(n))$\\
    \bf{MRCA Uncertainty}\\
    $\mathcal{O}(m)$
  }
  \makecell[{{p{0.14\textwidth}}}]{
  \raggedright
    where $m$ is gens since MRCA and $n$ is total gens elapsed.
  }\\\hline
    \adjustbox{
      minipage=10em,
      rotate=90,
    }{
      \centering
      \textbf{Geometric Sequence\\$n$th Root}
      \par
    }
  &
    \makecell{
      \includegraphics[valign=t,width=0.3\textwidth]{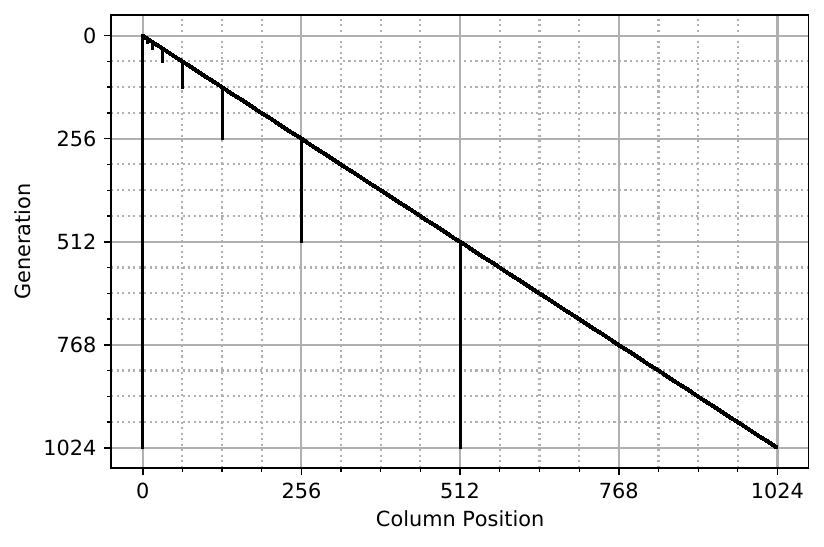}
    }
  &
    \makecell{
      \includegraphics[valign=t,width=0.3\textwidth]{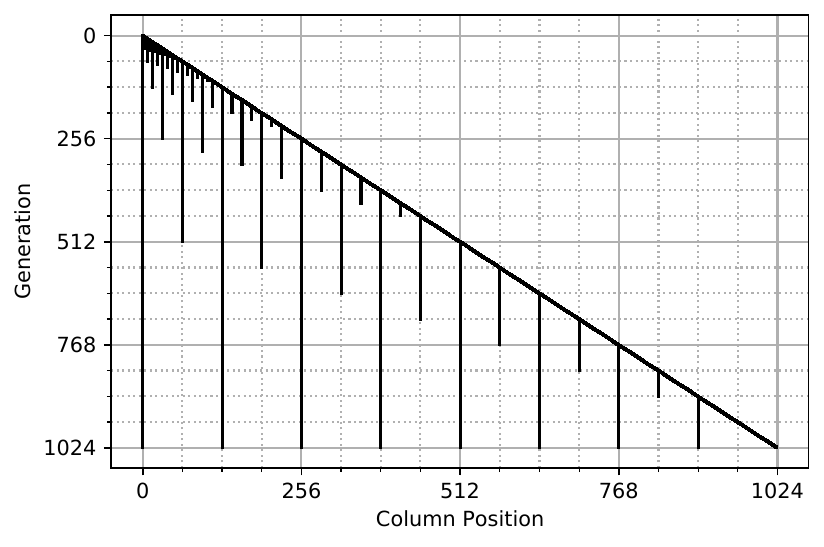}
    }
  &
  \makecell[{{p{0.14\textwidth}}}]{
  \centering
    \bf{Space Complexity}\\
    $\mathcal{O}(1)$\\
    \bf{MRCA Uncertainty}\\
    $\mathcal{O}(n^{1/a} \times m)$
  }
  \makecell[{{p{0.14\textwidth}}}]{
  \raggedright
    where $m$ is gens since MRCA, $a$ is a user-chosen constant, and $n$ is total gens elapsed.
  }\\\hline
    \adjustbox{
      minipage=10em,
    rotate=90,
    }{
      \centering
      \textbf{Curbed Recency-proportional Resolution}
      \par
    }
  &
    \makecell{
      \includegraphics[valign=t,width=0.3\textwidth]{hereditary-stratigraph-concept/tex/submodules/hereditary-stratigraph-concept-binder/binder/retention-policies/teeplots/guaranteed_depth_proportional_resolution=1+num_layers=1024+stratum_retention_predicate=tapered-depth-proportional-resolution+viz=tweaked-stratum-retention-drip-plot+ext=}
    }
  &
    \makecell{
      \includegraphics[valign=t,width=0.3\textwidth]{hereditary-stratigraph-concept/tex/submodules/hereditary-stratigraph-concept-binder/binder/retention-policies/teeplots/guaranteed_depth_proportional_resolution=4+num_layers=1024+stratum_retention_predicate=tapered-depth-proportional-resolution+viz=tweaked-stratum-retention-drip-plot+ext=}
    }
  &
  \makecell[{{p{0.14\textwidth}}}]{
  \centering
    \bf{Space Complexity}\\
    $\mathcal{O}(1)$\\
    \bf{MRCA Uncertainty}\\
    $\mathcal{O}(n^{1/a} \times m)$
  }
  \makecell[{{p{0.14\textwidth}}}]{
  \raggedright
   where $m$ is gens since MRCA, $a$ is a user-chosen constant, and $n$ is total gens elapsed.
  }

  \end{tabular}
  \caption{
  Comparison of stream curation policy algorithms.
  Policy visualizations show retained strata in black.
  Time progresses along the $y$-axis from top to bottom.
  New strata are introduced along the diagonal and then ``drip'' downward as a vertical line until eliminated.
  The set of retained strata present within a column at a particular generation $g$ can be read as intersections of retained vertical lines with a horizontal line with intercept $g$.
  Policy visualizations are provided for two contrastive parameterizations for each policy algorithm.
  }
  \label{fig:retention-policies}
\end{figure*}

This section collects specification and validation of five proposed stream curation algorithms.
These algorithms operate online on a rolling stream of incoming observations to maintain a representative subsample of retained observations.
Stream curation algorithms differ in the growth rate allowed for the curated observation collection and with respect to relative prioritization of retaining recent observations compared to older observations.


We introduce the following five curation policy algorithms,
\begin{itemize}
\item Fixed-Resolution (FR) Policy Algorithm (Section \ref{sec:fixed-resolution-algo}),
\item Depth-Proportional Resolution (DPR) Policy Algorithm (Section \ref{sec:depth-proportional-resolution-algo}),
\item Recency-Proportional Resolution (RPR) Policy Algorithm (Section \ref{sec:recency-proportional-resolution-algo}),
\item Geometric Sequence $n$th Root(GSNR) Policy Algorithm (Section \ref{sec:geom-seq-nth-root-algo}), and
\item Curbed Recency-Proportional Resolution (CRPR) Policy Algorithm (Section \ref{sec:curbed-recency-proportional-resolution-algo}).
\end{itemize}
The accomanying \texttt{hstrat} library provides reference implementations for all five policy algorithms \citep{moreno2022hstrat}.

Figure \ref{fig:retention-policies} compares retention patterns induced by each algorithm and recaps each policy algorithm's principal properties.
FR and DPR follow even curation prioritization while RPR, GSNR, and CRPR follow recency-proportional curation prioritization.
Recency-proportional techniques are potentially useful in a variety of contexts where recent information is more valuable or varied.
Collection size grows the most aggressively under FR as $\mathcal{O}(n)$.
RPR reduces collection size growth to $\mathcal{O}(\log n)$.
The remaining algorithms enforce a fixed cap on curated collection size.
Note that GSNR and CRPR exhibit identical asymptotic properties.
We include both, as CRPR is an engineered extension of GSNR that improves the efficacy of available space usage during initial shallow record depth.

Appropriate algorithm choice will depend on use case scenario.
Relevant criteria to consider include
\begin{itemize}
  \item uncertainty and magnitude of upper bounds on record depth, if any,
  \item available storage capacity,
  \item relative importance of recent and ancient observations, and
  \item any hard record quality requirements (i.e., maximum acceptable gap size).
\end{itemize}

In a real-time scenario, record depth bounds would be considered in terms of upper bounds on chronological duration of record collection and the real-time sampling rate of observations.

\subsection{Fixed Resolution (FR) Policy Algorithm}
\label{sec:fixed-resolution-algo}

The fixed resolution (FR) policy algorithm adopts a simplistic strategy: retain observation time points at intervals of a user-specified factor $r$.
This procedure equates to naive downsampling via decimation \citep[p. 31]{crochiere1983multirate}.
We include discussion of this policy algorithm primarily for completeness.

The procedures to decide eliminated time points during the update process and to enumerate retained observation times are trivial.
Pruning targets time points that are not multiples of the fixed resolution $r$ and enumeration traverses time points at an even stride $r$ until reaching record depth.

Unsurprisingly, extant record size order of growth is $\mathcal{O}(n)$.
A simple justification can be given: record depth $n$ provides an upper bound for extant record size because, being strictly subtractive, pruned size cannot exceed record depth.
Likewise, the downsampling factor $r$ enforces an upper bound $r$ on gap size between retained observations.
Any larger gap size would require at least one time point multiple of $r$ to have been discarded.

The FR policy algorithm provides stable absolute accuracy indefinitely.
Such an approach will be necessary for scenarios that tolerate only rigid observational uncertainties.
However, because extant record size grows linearly, this policy does not suit applications expecting long observational duration, high-frequency observation, or storage limitation.
Figure \ref{fig:retention-policies} includes a time-lapse of the extant record under the FR policy algorithm.


\subsection{Depth-Proportional Resolution (DPR) Policy Algorithm}
\label{sec:depth-proportional-resolution-algo}

The depth-proportional resolution policy algorithm provides capped extant record size with even coverage over record history.
This guarantee requires retained observations to be spaced with a gap width proportional to record depth.
Alternatively, DPR can be seen as interspersing the historical record with a fixed number of waypoints.

Because observation time points are immutable after the fact, translating this naive DPR plan to a rolling, ``online'' basis necessitates a further consideration.
To conservatively maintain resolution guarantees, it is acceptable to err on the side of caution by choosing gap sizes smaller than the worst-case requirement.
This approach allows a simple trick for achieving policy self-consistency: flooring gap sizes to the next lower power of two.
Under this scheme, gap size will periodically double.
Beacause multiples of a binary power superset multiples of higher binary power, self-consistency is maintained.
As intuition, therefore, the full DPR policy algorithm can be conceived of through a simple principle: 
each time a capacity threshold is reached, every second observation is eliminated.

Policy algorithm behavior is parametrized by a minimum number of bin windows over record history, $r$.
All gap sizes are equal (or halved), so absolute resolution guarantee of at least $n/r$, with $n$ as record depth, applies.
Further, because binary flooring operations at most halve gap widths, record count at most doubles.
This property gives the record size cap of $2r$.

Algorithm \ref{alg:depth-proportional-resolution-algo-enum-retained-ranks} provides enumeration of retained time points under the DPR scheme.
Although this process can be achieved via set subtraction between enumerations at successive time points with $\mathcal{O}(1)$ complexity, Algorithm \ref{alg:depth-proportional-resolution-algo-gen-drop-ranks} provides a more expedient approach.
Figure \ref{fig:retention-policies} includes a time-lapse of the extant record under the DPR policy algorithm.

For simplicity, we have presented a bare-bones approach to depth-proportional resolution, where the entire record is simultaneously decimated by a factor of two upon reaching capacity.
This procedure results in regular episodes where extant record count instantaneously halves.
Such fluctuation may be undesirable.
Many use-cases for constant space complexity will arise from fixed memory allocation.
Such reserved memory cannot typically be used for other purposes, so any unused space would be wasted.

An alternate ``tapered'' variant of the depth-proportional resolution algorithm remedies this space-usage quirk.
The tapered approach eliminates phased-out observations one by one as new observations accrue, but otherwise has the same properties as the algorithms described for DPR.
The accompanying \texttt{hstrat} software library implements both variants.

\begin{algorithm}
\caption{Depth-proportional Resolution Stratum Enumeration}
\label{alg:depth-proportional-resolution-algo-enum-retained-ranks}
\begin{algorithmic}[1]
    \Require{ $\texttt{n}$ -- the number of strata deposited }
    \Require{ $\texttt{r}$ -- the fixed resolution desired }
    \Ensure{ array of retained strata }

    \State $\texttt{uncertainty} \gets (\text{largest integral power of two} \le \frac{n}{r + 1}) \lor 1$
    \State $\texttt{arr} \gets \text{empty array of length } \frac{n}{\texttt{uncertainty}} + 1$

    \For{$i = 0$ \textbf{to} $\texttt{n} - 1$ \textbf{step} $\text{uncertainty}$}
        \State $\texttt{arr} [$i$] \gets i$
    \EndFor
    \State $\texttt{last\_rank} \gets \texttt{n} - 1$
    \If{$\texttt{last\_rank} > 0$ \textbf{and} $\texttt{last\_rank} \bmod \texttt{uncertainty} \neq 0$}
        \State $\texttt{\texttt{last\_rank}} [$i$] \gets \texttt{last\_rank}$
    \EndIf
\end{algorithmic}
\end{algorithm}

\begin{algorithm}
\caption{Depth-proportional Resolution Discard Generator}
\label{alg:depth-proportional-resolution-algo-gen-drop-ranks}
\begin{algorithmic}[1]
    \Require{ $\texttt{n}$ -- the number of strata deposited }
    \Require{ $\texttt{r}$ -- the fixed resolution desired }
    \Ensure{ array of dropped strata }

    \State $\texttt{curr\_uncertainty} \gets (\text{largest integral power of two} \le \frac{n}{r + 2}) \lor 1$
    \State $\texttt{prev\_uncertainty} \gets (\text{largest integral power of two} \le \frac{n}{r + 1}) \lor 1$
    \State $\texttt{arr} \gets \text{empty array of length } \frac{n}{\texttt{uncertainty}} + 1$

    \If{$\texttt{curr\_uncertainty} \neq \texttt{prev\_uncertainty}$}
        \For{$i = \texttt{prev\_uncertainty}$ \textbf{to} $\texttt{n} - 2$ \textbf{step} $\text{curr\_uncertainty}$}
            \State $\texttt{arr} [$i$] \gets i$
        \EndFor
    \EndIf
    \If{$\texttt{n} - 2 \bmod \texttt{curr\_uncertainty} \neq 0$}
        \State $\texttt{\texttt{last\_rank}} [\texttt{n} - 2] \gets \texttt{n} - 2$
    \EndIf
\end{algorithmic}
\end{algorithm}

\subsection{Recency-proportional Resolution (RPR) Policy Algorithm}
\label{sec:recency-proportional-resolution-algo}

This stream curation algorithm's properties fall between the properties of the fixed resolution (FR) and depth-proportional resolution (DPR) policy algorithms, covered in the immediately preceding sections.

Recall that the DPR policy algorithm's gap widths grow in linear proportion to record depth.
In contrast, the fixed resolution algorithm's gap width remains constant below a specified bound across record depths.

The recency-proportional resolution (RPR) policy algorithm bounds gap width to a linear factor of layer age (i.e., time steps back from the newest layer).
Here, layer refers to the $m$th observation ingressed from the underlying data stream being curated.

Suppose $n$ data stream observations have elapsed.
Then, for a user-specified constant $r$, no gap width for layer $m$ will exceed size
\begin{align}
  \left\lfloor \frac{n - m}{r} \right\rfloor.
  \label{eqn:rpr-gap}
\end{align}
Resolution at each layer widens linearly with record depth.
Consequently, resolution widens in linear proportion to layer age.
Resolution for any given layer age, however, remains constant for all record depths.

The FR and DPR policy algorithms exhibit $O(r)$ and $O(rn)$ extant record orders of growth, respectively.
We will show extant record order of growth as $O(r\log{n})$ under the recency-proportional resolution policy algorithm.

Algorithm \ref{alg:recency-proportional-algo-gen-drop-ranks} enumerates time points of dropped observations under the RPR policy algorithm.
Figure \ref{fig:retention-policies} includes a time-lapse of the extant record under this policy algorithm.

The extant record is determined iteratively, beginning at observation time zero --- which is always retained.
Per Equation \ref{eqn:rpr-gap}, gap width to the next retained observation can be at most $\lfloor n/r \rfloor$ sites, where $n$ is record depth.
Although retaining the observation at time $\lfloor n/r \rfloor$ would satisfy policy resolution guarantees, a slight complication is necessary to ensure self-consistency.
A fuller rationale will follow, but in short, gap width is floored to the next lower power of two,
\begin{align*}
  2^{\lfloor \log_{2}\left(\frac{n}{r}\right) \rfloor}.
\end{align*}
The next iteration repeats the procedure from the newly retained observation time instead of from time zero.
Iteration continues until reaching the newest observation.

The set of observations to eliminate can be calculated from set subtraction between enumerations of the historical record at time points $t-1$ and $t$.
So, update time complexity follows from extant record enumeration time complexity, which turns out to be $O(\log n)$.
We provide a tested, but unproven, constant-time pruning enumeration implementation in the \texttt{hstrat} library, but will not cover it here. 
The extant record order of growth of $O(\log n)$ also follows from the record enumeration algorithm, as detailed in Theorem \ref{thm:recency-proportional-resolution-algo-space-complexity}.


Why does flooring step sizes to a binary power ensure self-consistency?
Let us begin by noting properties applicable to all layers $l$,
\begin{enumerate}
\item gap width provided at retained layer $l$ increases monotonically as record depth grows,
\item the retained observation preceding or at $l$ has observation time at an even multiple of surrounding gap widths, and
\item all observations at time points that are multiples of gap width past $l$ up to the newest observation are retained.
\end{enumerate}
Observe that gap width decreases monotonically with decreasing layer age (i.e., increasing layer recency).

Properties 2 and 3 occur as a result of stacking monotonically-decreasing powers of two.
Subsequent smaller powers of two tile evenly to all multiples of a larger power of 2, giving property 3.
Conversely, preceding larger powers of 2 can be evenly divided by succeeding smaller powers of 2, ensuring that the edges of smaller powers of 2 gaps occur at even multiples of their gap width, giving property 2.

Under the binary flooring procedure, when gap size increases at a layer it will double (or quadruple, octuple, etc.).
Availability of the new gap endpoint after a gap size increase is guaranteed from the tiling properties due to that endpoint being an even multiple of original step size.

The RPR policy algorithm provides stable relative accuracy indefinitely.
This makes it particularly attractive in applications to phylogenetic tracking scenarios using hereditary stratigraphy.
To meaningfully describe an ancestry tree with deep branches, information must be retained across all evolutionary time but higher absolute estimation error is typically acceptable in describing more ancient most recent common ancestor (MRCA) events.%
\footnote{%
At comparable annotation sizes, we have found that recency-proportional distribution of gap widths outperforms even gap width distribution in phylogenetic information recovery \citep{moreno2022hereditary}.
Preliminarily, maintaining 3\% relative precision appears sufficient to eliminate most bias from reconstruction error on phylogenetic metrics \citep{moreno2023toward}.
}

The RPR policy algorithm's indefinite stability may be particularly useful in scenarios of indefinite or indeterminate record keeping duration.
Although annotation extant record size grows unboundedly, logarithmic memory usage growth is manageable in most practical scenarios.
However, this policy would not suit applications with hard caps on annotation size. 

\begin{algorithm}
\caption{Recency-proportional Resolution Stratum Discard Generator}
\label{alg:recency-proportional-algo-gen-drop-ranks}
\begin{algorithmic}
    \Require{ $\texttt{n}$ -- the number of strata deposited }
    \Require{ $\texttt{r}$ -- the fixed resolution desired }
    \Ensure{ array of dropped strata }

    \Procedure{NumberToCondemn}{$\texttt{n}, \texttt{r}$}
        \If{$(\texttt{n} \bmod 2 = 1) \lor (\texttt{n} < 2 \cdot \texttt{r} + 1)$}
            \Return $0$
        \Else
            \Return $1 + \Call{NumberToCondemn}{$\texttt{n} / 2$, \texttt{r}}$
        \EndIf
    \EndProcedure

    \State $\texttt{num\_to\_condemn} \gets \Call{NumberToCondemn}{\texttt{n}, \texttt{r}}$
    \State $\texttt{arr} \gets \text{empty array of length num\_to\_condemn}$

    \For{$i = 0$ \textbf{to} $\texttt{num\_to\_condemn} - 1$}
        \State $\texttt{arr} [$i$] \gets \texttt{n} - 2^{i} \cdot (2 \texttt{r} + 1)$
    \EndFor
\end{algorithmic}
\end{algorithm}

\begin{theorem}{Recency-proportional Resolution Space Complexity}
\label{thm:recency-proportional-resolution-algo-space-complexity}

The \gls{extant record size} of the Recency-proportional Resolution Policy Algorithm grows with order $\mathcal{\theta}{(k \log n)}.$

\end{theorem}

\begin{proof}
\label{prf:recency-proportional-resolution-algo-space-complexity}
As per \ref{sec:extant_record_oog}, we will set out to prove that output array of this policy algorithm has an order of growth of $\mathcal{\theta}{(k \log n)},$ where $k$ is a user-provided resolution and $n$ is the number of depositions.

Algorithm \ref{alg:recency-proportional-algo-gen-drop-ranks} determines the array of strata to be dropped at any given time point.
Observe that whenever $2 \mid n$ at least one stratum will be dropped.
More generally, for any positive integer $i \le \log_2 n$, we have that if $2^i \mid n$ then $i$ strata will be dropped. 
Thus, the number of dropped strata is bounded above by $\sum_{i=1}^{\log_2 n} n = n \log_2 n.$
As such, the number of retained strata is bound by $n - n \log_2 n \le \log_2 n$ for all positive $n.$
Given that no strata are dropped when $\frac{n}{2} - 1 < k,$ we observe that the output array of this policy algorithm is bound above by $\mathcal{O}{(k \log n)}.$
Via \ref{sec:extant_record_oog}, we can conclude that this bound is actually $\mathcal{\theta}{(k \log n)}.$
\end{proof}

\subsection{Geometric Sequence $n$th Root (GSNR) Policy Algorithm}
\label{sec:geom-seq-nth-root-algo}

The geometric sequence $n$th root (GSNR) policy algorithm arranges recency-proportional gap sizes among a capped-size set of retained observations.
Although recency-proportional gap size will not be bounded to a fixed threshold in this context, GSNR seeks to minimize worst relative gap size as much as possible.

Recall that the recency-proportional policy algorithm (RPR) exhibits logarithmic growth in extant record size with respect to record depth $n$.
When an increased order of magnitude depth is reached, additional observations must be retained under the RPR algorithm.
For $a = \log_b(n)$, $a$ is proportional to extant record size.
Equivalently, $b^a = n$.
Under RPR, growth in extant record size can be roughly conceptualized as related to insufficiency of the base $b$ to reach $n$ within $a$ multiplicative steps.
Growth in $a$ --- i.e., additional multiplication by $b$ --- can be thought of in terms of adding a level of structural hierarchy within the layout of retained observations.
As $n$ increases, additional levels of structural hierarchy become necessary.
These additional hierarchical levels increase extant record size.

In order to prevent such unbounded growth, the GSNR policy algorithm fixes the number of hierarchical levels $a$ and accommodates additional record depth by adjusting the multiplicative factor $b$.
This scheme can be imagined enforcing as arrangement of $a$ exponentially-spaced target points along the historical record.
As time elapses, the quantity of target points remains constant.
The target points shift to fill $n$ by increasing their exponential spacing factor $b$.
The necessary magnitude of $b$ works out as $b = n^{1/a}$.
Target ages therefore correspond to $n^{0/a}, n^{1/a}, \ldots, n^{a/a}$.
This geometrically-spaced target point sequence eponymizes the GSNR policy algorithm.
Converting observation age to absolute time point, targets span $n - n^{0/a}, n - n^{1/a}, \ldots, n - n^{a/a}$.


Now, attention turns to exploiting the $n$th root geometric targets to define a retention policy.
We will break the problem down to consideration of one individual target point $n - n^{x/a}$.
Under the constraint of $\mathcal{O}(1)$ total space for curated observations, we can only curate a fixed number of observations per target point.
We will seek to curate a fixed size collection of retained observations to bound gap size past the target point below $n^{x/a}$.

By nature of definition, target point times advance monotonically.
As a consequence, a retained observation can remain behind a target point indefinitely.
We will incorporate such coverage into our design --- let's call such a point behind the target the ``backstop'' $\beta$.

We will use a power of 2 trick to maintain backstop coverage.
To begin, let us take the binary floor of half $n^{x/a}$,
\begin{align*}
  \kappa(n)
  &=
  2^{\lfloor \log_{2}(n^{x/a}/2) \rfloor}.
\end{align*}
We will retain recent time points that are multiples of this value $\kappa$.
Note that $n$ strictly increasing implies $\kappa(n)$ monotonically increasing.

Let's define a floor $B$ to help place our backstop $\beta$,
\begin{align*}
  B(n)
  &=
  \max \left(
    n - \left\lceil  \frac{3n^{x/a}}{2} \right\rceil,
    0
  \right)
\end{align*}
By design, $B$ precedes target point $n - n^{x/a}$.
Again, with $x < a$, $n$ strictly increasing implies $B(n)$ monotonically increasing.

Rounding $B$ up to the next time point aligned to cadence $\kappa$ gives our backstop time point $\beta$,
\begin{align*}
  \beta(n)
  &=
  B(n) + \big(-B(n) \bmod \kappa(n)\big).
\end{align*}
It can be shown that $\beta(n) \leq n^{x/a}$.
Because $B(n)$ is monotonically increasing, $\beta(n)$ is as well.

We will retain the time point set $S_x$ comprising multiples of $\kappa$ at or after the backstop $\beta$,
\begin{align*}
  S_x(n) = \{\, t \mid \beta(n) \leq t < n \text{ and } t \mod \kappa(n) = 0 \,\}.
\end{align*}

Why are time points $S_x(n)$ guaranteed to be a subset of $S_x(n-1) \cup \{n\}$ (i.e., self-consistency)?
Consider several non-mutually exclusive possible scenarios that could occur when transitioning from $n - 1$ to $n$,
\begin{enumerate}
  \item $\kappa$ changes: $\kappa(n - 1) \neq \kappa(n)$.

  Because $\kappa$ is monotonically increasing, $\kappa(n - 1) < \kappa(n)$.
  Binary flooring procedures have ensured $\kappa(n - 1)$ and $\kappa(n)$ are perfect powers of 2.
  Thus, $\kappa(n)$ is an even multiple of $\kappa(n - 1)$.
  So,
  \begin{align*}
    &\{\, t \mid \beta(n) \leq t < n \text{ and } t \mod \kappa(n) = 0 \,\}\\
    &\subseteq \{\, t \mid \beta(n) \leq t < n \text{ and } t \mod \kappa(n - 1) = 0 \,\}.
  \end{align*}
  Also, because $\beta$ monotonically increasing, $\beta{n - 1} \leq \beta{n}$.

  In conjunction, these stipulations give us $S_x(n) \subseteq \{S_x(n - 1), n\}$.

  \item $\beta$ changes: $\beta(n - 1) \neq \beta(n)$.

  Because $\beta$ is monotonically increasing, $\beta(n - 1) < \beta(n)$.
  We have $\beta(n) \mod \kappa(n) = 0$.
  Because $\kappa(n)$ is an even multiple of $\kappa(n - 1)$, we have $\beta(n) \mod \kappa(n - 1) = 0$.
  This implies $\beta(n) \in S_x(n - 1)$.
  Because $\beta$ monotonically increasing, change in $\beta$ strictly shrinks $S_x(n)$.

  \item $n \mod \kappa(n) = 0$

  Then the observation at time point $n$ is claimed for inclusion within $S_x(n)$.
  It is available, having just presently occurred.

\end{enumerate}

Note that if none of the above occur, then $S_x(n - 1) = S_x(n)$.
Any combination of the above maintains $S_x(n - 1) \in \{S_x(n), n\}$

Why does this construction for target $n^{x/a}$ satisfy our $\mathcal{O}(1)$ space complexity constraint?
It can be shown that $|S_x(n)| \leq 6$.
This stems from cadence $\kappa$ as the binary floor of half $n^{x/a}$ (at most a quartering reduction) and backstop $\beta$ set at most $3/2 \times n^{x/a}$ time points back.






Why does this curated set for target $n - n^{x/a}$ satisfy our gap size bound $n^{x/a}$?
Because cadence $\kappa \leq \frac{n^{x/a}}{2}$, gap size satisfies the bound.
Because backstop $\beta \leq n - n ^{x/a}$, the target time point is covered within the cadenced range.

We bring curated sets for each target point together in a set union to produce the overall GSNR retained set $R$,
\begin{align*}
  R(n)
  &=
  \bigcup_{i=0}^{a} S_i(n).
\end{align*}
Note that under this construction, policy algorithm self-consistency and extant record size bounds follow from those shown for constructions for individual targets $n^{x/a}$.
\footnote{
A strict upper bound of $6a + 2$ for extant record size can be calculated, although we do not demonstrate it here.
}
Time points that should be dropped to enact the GSNR policy algorithm follow from set subtraction between $R(n)$ and $R(n+1)$.

We have discussed resolution guarantees for individual target point constructions, but what resolution guarantee is afforded overall for an arbitrary time point $t$ with age $g = n - t$?
For such a time point $t$, the tightest resolution guarantee is that of the next older target point.
Taking
\begin{align*}
\alpha = n^{1/a}
\end{align*}
the age of the next older target time point will be
\begin{align*}
\alpha^{ \lceil \log_{\alpha} g \rceil },
\end{align*}

By the target time point resolution guarantee we established earlier, the gap size provided at this target time point is bounded by its age.

Observe that, at most, the next older target time point age will be a factor of $\alpha = n^{1/a}$ greater than $g$.
So, the worst case absolute provided gap size is
\begin{align*}
\alpha \times g = n^{1/a} \times g.
\end{align*}

Worst case recency-proportional gap size is therefore $n^{1/a}$.

Figure \ref{fig:retention-policies} includes a time-lapse of the extant record under the GSNR policy algorithm.
The distinguishing feature of the GSNR policy algorithm is it keeps recency-proportional gap sizes past the point where RPR policy would overflow a given record size bound.
This lends it to very-long duration applications.


\begin{algorithm}
\caption{Geometric Sequence $n$th Root Stratum Enumeration}
\label{alg:geom-seq-nth-root-enum-retained-ranks}
\begin{algorithmic}
    \Require{ $\texttt{n}$ -- the number of strata deposited }
    \Require{ $\texttt{p}$ -- the interspersal  }
    \Require{ $\texttt{d}$ -- the degree } 
    \Ensure{ set of retained strata }
\end{algorithmic}

\begin{algorithmic}[1]
    \State $\texttt{set} \gets {0, n - 1}$

    \For{$p = 1$ \textbf{to} $\texttt{d} + 1$}
        \State $\texttt{r} \gets n^{\frac{p}{d}}$

        \State $\texttt{k} \gets \text{largest integral power of two} \le \max \{\frac{r}{p}, 1 \}$ 

        \State $\texttt{c} \gets \max \{n - \ceil{r + \frac{r}{p}}, 0\}$
        \State $\texttt{a} \gets \texttt{c} - (\texttt{c} \bmod -\texttt{k})$

        \State $\texttt{set} \gets \texttt{set} \cup \{\texttt{a}, \texttt{a} + \texttt{k}, \texttt{a} + 2\texttt{k}, {a} + 3\texttt{k}, \dots, n\}$
    \EndFor
\end{algorithmic}
\end{algorithm}

\subsection{Curbed Recency-proportional Resolution (CRPR) Policy Algorithm}
\label{sec:curbed-recency-proportional-resolution-algo}

The curbed recency-proportional resolution (CRPR) policy algorithm exists to enhance in-practice utility of recency-proportional resolution.
This policy algorithm combines the geometric sequence $n$th root (GSNR) and recency-proportional resolution (RPR) algorithms to harness the strengths of both.
Figure \ref{fig:retention-policies} includes a time-lapse of the extant record under the CRPR policy algorithm.

The GSNR algorithm sustains best-effort recency-proportional resolution to asymptotic limits, but distributes retained observations less effectively than RPR in earlier periods.
This can result in higher realized worst-case recency-proportional gap size than necessary.
However, RPR only makes use of available storage space up to the retention density of its parameterized resolution.
Choosing a low resolution to provide support for long historical depths would cause the majority of available storage space to go unused earlier on.
Nevertheless, at even the lowest possible parameterized resolution, extant record size will eventually outgrow any size cap under RPR.

The CRPR policy algorithm delivers the best of both worlds: full, effective storage space use of high-resolution parameterized RPR policy plus the graceful, indefinite tail support of GSNR policy.
It achieves this balance by splicing these two policy algorithms together, transitioning from RPR to GSNR once the lowest-resolution RPR policy would exceed available space.
Before reaching that point of transition to GSNR, CRPR maximizes use of available space by cycling through a series of successively lower-resolution RPR parameterizations, downgrading each time usage would exceed available space.
The CRPR policy algorithm itself provides $\mathcal{O}(1)$ extant record size order of growth, and is parameterized by a desired upper bound $m$ on retained observation count.
Support is provided for $m \geq 8$.

The CRPR policy switches from RPR to GSNR at time point,
\begin{align*}
n = \left\lfloor \frac{2^{\lfloor m/3 \rfloor}}{2} \right\rfloor
\end{align*}

with the delegated-to GSNR policy algorithm permanently parameterized to degree:
\begin{align*}
a = \max \left(
  \left\lfloor \frac{m - 2}{6} \right\rfloor,
  1
\right).
\end{align*}

At preceding time points $n$, RPR policy is parameterized to resolution
\begin{align*}
r = \left(
  m \,
  \left\lfloor
  \frac{1}{\lceil \log_2(n + 1) \rceil + 1}
  \right\rfloor
  - 1
\right).
\end{align*}
Note that this resolution $r$ progressively decreases with record depth $n$.

In describing the CRPR algorithm, we will build off the properties of the RPR and GSNR algorithms established in sections \ref{sec:recency-proportional-resolution-algo} and \ref{sec:geom-seq-nth-root-algo}.
Under the CRPR, at any given point in time either the RPR or GSNR is currently active.
If RPR is active, resolution parameterization depends on record depth $n$.
Whichever algorithm is currently active enumerates retained observation time points.
Time points to delete can be calculated through set subtraction of sequential curated collection enumerations.
Because curated collection size is bounded, this procedure is $\mathcal{O}(1)$.
Asymptotic properties result solely from GSNR, as it constitutes the final destination of the CRPR stitched policy sequence.
Except when switching CRPR's active policy algorithm, we can also carry forward self-consistency assurances from the existing RPR and GSNR policy algorithms.
However, when switching CRPR's active policy algorithm, we do need to assess self-consistency.
Transitioned-to time point collections must be a subset of transitioned-from time point collections.
Self-consistency turns out to largely arise from peculiarities of stacking monotonically decreasing binary powers.

In application, the CRPR policy algorithm should be preferable to GSNR and RPR for nearly all scenarios that call for recency-proportional resolution under capped-size storage limitations.
Except when switching active policy, update implementation can be optimized by replacing the set subtraction procedure with the active policy algorithm's implementation.
It turns out such transitions exclusively occur at perfect power-of-two time points.

\section{Conclusion} \label{sec:conclusion}

In this work, we have characterized the stream curation problem --- how to maintain a temporally-representative running archive of stream data --- and provided several algorithms to solve it, each meeting different criteria.
We have systematized curatorial properties of data retention strategies, covering how many data items should be retained and how retained data items should distribute over past time.
We provide five policy algorithms that target a spectrum of size/coverage trade-offs and demonstrate procedures to enact them efficiently.
Implementations enable key optimizations to trim archive storage size overhead through efficient and computationally reducible (i.e., stateless) determination of retained contents.

Within the original context for their development, presented stream curation algorithms are key to tunability and efficiency of hereditary stratigraphy.
However, the stream curation problem generalizes beyond hereditary stratigraphy, and memory-smart optimizations provided here stand to boost the data stream mining capabilities of low-grade hardware in roles such as sensor nodes and data logging devices.

Much work remains.
We are particularly interested in exploring further adaptations to our approach to stream curation to improve efficacy of space use, efficiency of the update process, and simplicity of implementation for fixed-space contexts.
One promising direction involves adapting retention policies to directly specify a buffer index position to replace.
This would remove complications of shuffling down entries when early buffer entries are removed and, after initial population of the buffer, would ensure completely full space utilization.
Plenty remains for theoretical analysis, as well.
In particular, the extent to which stream curation policy algorithms minimize the number of records necessary to achieve their guarantees should be considered.

Ultimately, however, our interest in stream curation is application-driven.
To these ends, reference implementations of most algorithms described are provided in the public-facing \texttt{hstrat} Python library \citep{moreno2022hstrat}.
In addition to hereditary stratigraphy-specific tools, the \texttt{hstrat} library makes stream curation algorithm implementations available via a few lines of code.

\putbib

\end{bibunit}

\clearpage
\newpage

\section*{Appendix}

\setcounter{section}{0}
\makeatletter
\renewcommand \thesection{S\@arabic\c@section}
\renewcommand\thetable{S\@arabic\c@table}
\renewcommand \thefigure{S\@arabic\c@figure}
\makeatother

\begin{bibunit}

\printunsrtglossary[numberedsection=autolabel]

\section{Extant Record Size Order of Growth of Retention Policy Algorithms} \label{sec:extant_record_oog}
The concept of `extant record order of growth' is defined to be the upper bound of memory usage of a policy algorithm with respect to elapsed generations.
Since policy algorithms determine the strata retained into the following generation without the use of any additional data structures, the upper bound of their memory usage will be equivalent in order to the number of retained strata.
In fact, given that no intermediate data structures other than the final output array are used, we can guarantee the lower bound will also be equivalent.

As such, Extant Record Size Order of Growth proofs will set out to prove that the number of retained strata of each respective retention policy algorithm will grow with an order of $\mathcal{\theta}(f(n, k)),$ where $f$ is unique to each retention policy algorithm, $k$ is a user-defined constant and $n$ is the number of strata in the current generation.

\putbib

\end{bibunit}

\end{document}